\documentclass{elsarticle}

\usepackage[colorlinks=true,urlcolor=webbrown,linkcolor=RoyalBlue,citecolor=webgreen,pdfhighlight=/O]{hyperref}

\usepackage{graphicx,amssymb,amsmath,amsthm}
\usepackage{subcaption}
\usepackage{xspace}

\journal{Computational Geometry}

\graphicspath{{figures/}}

\newcommand{\etal}{et al.\xspace}

\newcommand{\vl}{\ensuremath{v_{\textsc{l}}}}
\newcommand{\vr}{\ensuremath{v_{\textsc{r}}}}
\newcommand{\free}{\ensuremath{\mathcal{F}}}
\newcommand{\nf}{\ensuremath{\overline{F}}}

\newtheorem{definition}{Definition}
\newtheorem{theorem}[definition]{Theorem}
\newtheorem{lemma}[definition]{Lemma}

\newtheorem{corollary}[definition]{Corollary}

\begin{document}

\begin{frontmatter}

\title{Flips in Edge-Labelled Pseudo-Triangulations\tnoteref{mytitlenote}}
\tnotetext[mytitlenote]{The results in this paper appeared in CCCG as an extended abstract, and are part of the second author's PhD thesis. This work was partially supported by NSERC.}

\author{Prosenjit Bose} \ead{jit@scs.carleton.ca}
\author{Sander Verdonschot} \ead{sander@cg.scs.carleton.ca}
\address{School of Computer Science, Carleton University, Ottawa}

\begin{abstract}
 We show that $O(n^2)$ exchanging flips suffice to transform any edge-labelled pointed pseudo-triangulation into any other with the same set of labels. By using insertion, deletion and exchanging flips, we can transform any edge-labelled pseudo-triangulation into any other with $O(n \log c + h \log h)$ flips, where $c$ is the number of convex layers and $h$ is the number of points on the convex hull.
\end{abstract}


\end{frontmatter}

\section{Introduction}

A \emph{pseudo-triangle} is a simple polygon with three convex interior angles, called \emph{corners}, that are connected by reflex chains. Given a set $P$ of $n$ points in the plane, a \emph{pseu\-do-tri\-an\-gu\-la\-tion} of $P$ is a subdivision of its convex hull into pseudo-triangles, using all points of $P$ as vertices (see~Figure~\ref{fig:el-pt}). A pseu\-do-tri\-an\-gu\-la\-tion is \emph{pointed} if all vertices are incident to a reflex angle in some face (including the outer face; see Figure~\ref{fig:el-pointed-pt} for an example). Pseu\-do-tri\-an\-gu\-la\-tions find applications in areas such as kinetic data structures~\cite{kirkpatrick2002kinetic} and rigidity theory~\cite{streinu2005pseudo}. More information on pseu\-do-tri\-an\-gu\-la\-tions can be found in a survey by Rote, Santos, and Streinu~\cite{rote2007pseudotriangulations}.

\begin{figure}[htb]
 \centering
 \begin{subfigure}[b]{0.48\textwidth}
  \centering
  \includegraphics{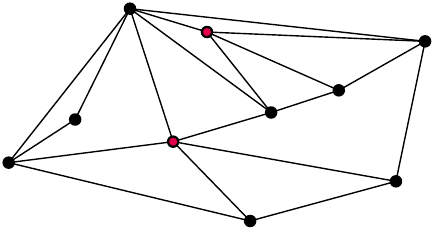}
  \caption{}
  \label{fig:el-pt}
 \end{subfigure}
 \begin{subfigure}[b]{0.48\textwidth}
  \centering
  \includegraphics{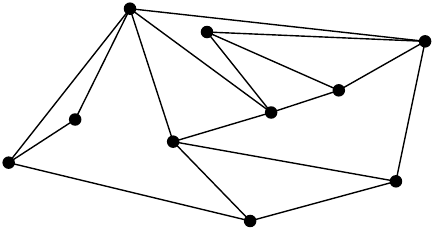}
  \caption{}
  \label{fig:el-pointed-pt}
 \end{subfigure}
 \caption{(a) A pseu\-do-tri\-an\-gu\-la\-tion with two non-pointed vertices. (b) A pointed pseu\-do-tri\-an\-gu\-la\-tion.}
 \vspace{-1em}
\end{figure}

Since a regular triangle is also a pseudo-triangle, pseu\-do-tri\-an\-gu\-la\-tions generalize triangulations (subdivisions of the convex hull into triangles). In a triangulation, a flip is a local transformation that removes one edge, leaving an empty quadrilateral, and inserts the other diagonal of that quadrilateral. Note that this is only allowed if the quadrilateral is convex, otherwise the resulting graph would be non-plane. Lawson~\cite{lawson1972transforming} showed that any triangulation with $n$ vertices can be transformed into any other with $O(n^2)$ flips, and Hurtado, Noy, and Urrutia~\cite{hurtado1999flipping} gave a matching $\Omega(n^2)$ lower bound.

\begin{figure}[htb]
 \centering
 \includegraphics{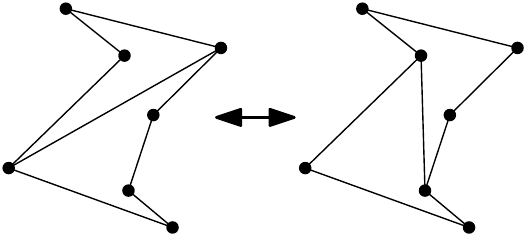}
 \caption{A flip in a pseudo-quadrilateral.}
 \label{fig:el-flip}
\end{figure}

Pointed pseu\-do-tri\-an\-gu\-la\-tions support a similar type of flip, but before we can introduce this, we need to generalize the concept of pseudo-triangles to \emph{pseudo-$k$-gons}: weakly simple polygons\footnote{A weakly simple polygon is a plane graph with a bounded face that is incident to all edges.} with $k$ convex interior angles. A diagonal of a pseudo-$k$-gon is called a \emph{bitangent} if the pseudo-$k$-gon remains pointed after insertion of the diagonal. In a pointed pseu\-do-tri\-an\-gu\-la\-tion, \emph{flipping} an edge removes the edge, leaving a pseudo-quadrilateral, and inserts the unique other bitangent of the pseudo-quadrilateral (see Figure~\ref{fig:el-flip}). In contrast with triangulations, all internal edges of a pointed pseu\-do-tri\-an\-gu\-la\-tion are flippable. Bereg~\cite{bereg2004transforming} showed that $O(n \log n)$ flips suffice to transform any pointed pseu\-do-tri\-an\-gu\-la\-tion into any other.

Aichholzer~\etal~\cite{aichholzer2003pseudotriangulations} showed that the same result holds for all pseu\-do-tri\-an\-gu\-la\-tions (including triangulations) if we allow two more types of flips: \emph{insertion} and \emph{deletion} flips. As the name implies, these either insert or delete one edge, provided that the result is still a pseu\-do-tri\-an\-gu\-la\-tion. To disambiguate, they call the other flips \emph{exchanging} flips. In a later paper, this bound was refined to $O(n \log c)$~\cite{aichholzer2006transforming}, where $c$ is the number of convex layers of the point set.

There is recent interest in \emph{edge-labelled} triangulations: triangulations where each edge has a unique label, and flips reassign the label of the flipped edge to the new edge. Araujo-Pardo~\etal~\cite{araujo2015colorful} studied the flip graph of edge-labelled triangulations of a convex polygon, proving that it is still connected and in fact covers the regular flip graph. They also fully characterized its automorphism group. Independently, Bose~\etal~\cite{bose2013flipping} showed that it has diameter $\Theta(n \log n)$.

Bose~\etal also considered edge-labelled versions of other types of triangulations. In many cases, the flip graph is disconnected. For example, it is easy to create a triangulation on a set of points with an edge that can never be flipped. Two edge-labelled triangulations in which such an edge has different labels can therefore never be transformed into each other with flips. On the other hand, Bose~\etal showed that the flip graph of edge-labelled combinatorial triangulations is connected and has diameter $\Theta(n \log n)$.

In this paper, we investigate flips in \emph{edge-la\-belled pseu\-do-tri\-an\-gu\-la\-tions}: pseu\-do-tri\-an\-gu\-la\-tions where each internal edge has a unique label in $\{1, \ldots, 3n - 3 - 2h\}$, where $h$ is the number of vertices on the convex hull ($3n - 3 - 2h$ is the number of internal edges in a triangulation). In the case of an exchanging flip, the new edge receives the label of the old edge. For a deletion flip, the edge and its label are removed, and for an insertion flip, the new edge receives an unused label from the set of all possible labels.

In contrast with the possibly disconnected flip graph of edge-labelled geometric triangulations, we show that $O(n^2)$ exchanging flips suffice to transform any edge-la\-belled pointed pseu\-do-tri\-an\-gu\-la\-tion into any other with the same set of labels. By using all three types of flips -- insertion, deletion and exchanging -- we can transform any edge-la\-belled pseu\-do-tri\-an\-gu\-la\-tion into any other with $O(n \log c + h \log h)$ flips, where $c$ is the number of convex layers of the point set, and $h$ is the number of vertices on the convex hull. In both settings, we have an $\Omega(n \log n)$ lower bound, making the second result tight.

\section{Pointed pseu\-do-tri\-an\-gu\-la\-tions}
\label{sec:el-pointed}

In this section, we show that every edge-la\-belled pointed pseu\-do-tri\-an\-gu\-la\-tion can be transformed into any other with the same set of labels by $O(n^2)$ exchanging flips. We do this by showing how to transform a given edge-la\-belled pointed pseu\-do-tri\-an\-gu\-la\-tion into a \emph{canonical} one. The result then follows by the reversibility of flips. 

\begin{figure}[htb]
 \centering
 \includegraphics{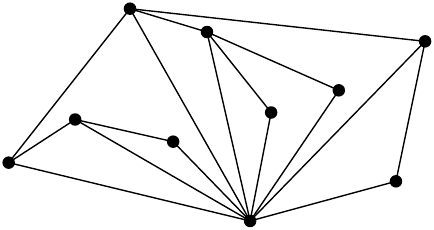}
 \caption{A left-shelling pseu\-do-tri\-an\-gu\-la\-tion.}
 \label{fig:el-left-shelling}
\end{figure}

Before we can start the proof, we need a few more definitions. Given a set of points in the plane, let $v_0$ be the point with the lowest $y$-coordinate, and let $v_1, \ldots, v_n$ be the other points in clockwise order around $v_0$. The \emph{left-shelling} pseu\-do-tri\-an\-gu\-la\-tion is the union of the convex hulls of $v_0, \ldots, v_i$, for all $2 \leq i \leq n$ (see Figure~\ref{fig:el-left-shelling}). Thus, every vertex after $v_1$ is associated with two edges: a \emph{bottom} edge connecting it to $v_0$ and a \emph{top} edge that is tangent to the convex hull of the earlier vertices. The \emph{right-shelling} pseu\-do-tri\-an\-gu\-la\-tion is similar, with the vertices added in counter-clockwise order instead.

As canonical pseu\-do-tri\-an\-gu\-la\-tion, we use the left-shelling pseu\-do-tri\-an\-gu\-la\-tion, with the bottom edges labelled in clockwise order around $v_0$, followed by the internal top edges in the same order (based on their associated vertex). Since we can transform any pointed pseu\-do-tri\-an\-gu\-la\-tion into the left-shelling pseu\-do-tri\-an\-gu\-la\-tion with $O(n \log n)$ flips~\cite{bereg2004transforming}, the main part of the proof lies in reordering the labels of a left-shelling pseu\-do-tri\-an\-gu\-la\-tion. We use two tools for this, called a \emph{sweep} and a \emph{shuffle}, that are each implemented by a sequence of flips. A sweep interchanges the labels of some internal top edges with their respective bottom edges, while a shuffle permutes the labels on all bottom edges.

\begin{lemma}
\label{lem:el-sort-left-shelling}
We can transform any left-shelling pseu\-do-tri\-an\-gu\-la\-tion into the canonical one with $O(1)$ shuffle and sweep operations.
\end{lemma}
\begin{proof}
In the canonical pseu\-do-tri\-an\-gu\-la\-tion, we call the labels assigned to bottom edges \emph{low}, and the labels assigned to top edges \emph{high}. In the first step, we use a shuffle to line up every bottom edge with a high label with a top edge with a low label. Then we exchange these pairs of labels with a sweep. Now all bottom edges have low labels and all top edges have high labels, so all that is left is to sort the labels. We can sort the low labels with a second shuffle. To sort the high labels, we sweep them to the bottom edges, shuffle to sort them there, then sweep them back.
\end{proof}

The remainder of this section describes how to perform a sweep and a shuffle with flips.

\begin{lemma}
\label{lem:el-degree-2-swap}
We can interchange the labels of the edges incident to an internal vertex $v$ of degree two with three exchanging flips.
\end{lemma}
\begin{proof}
Consider what happens when we remove $v$. Deleting one of its edges leaves a pseudo-quadrilateral. Removing the second edge then either merges two corners into one, or removes one corner, leaving a pseudo-triangle $T$. There are three bitangents that connect $v$ to $T$, each corresponding to the geodesic between $v$ and a corner of $T$. Any choice of two of these bitangents results in a pointed pseu\-do-tri\-an\-gu\-la\-tion. When one of them is flipped, the only new edge that can be inserted so that the result is still a pointed pseu\-do-tri\-an\-gu\-la\-tion is the bitangent that was not there before the flip. Thus, we can interchange the labels with three flips (see Figure~\ref{fig:el-degree-2-swap}).
\end{proof}

\begin{figure}[htb]
 \centering
 \includegraphics{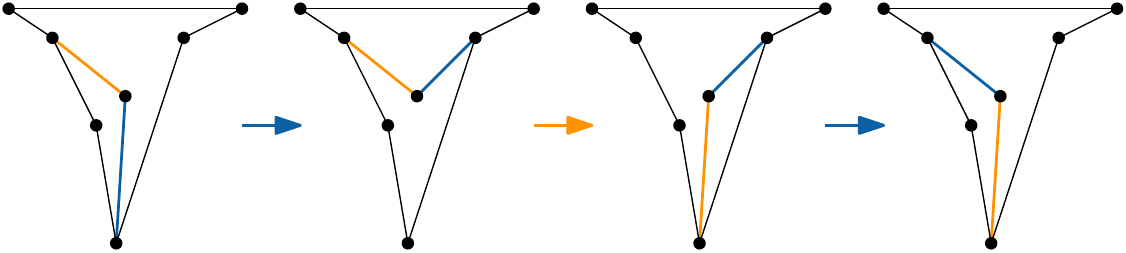}
 \caption{Interchanging the labels of the edges incident to a vertex of degree two.}
 \label{fig:el-degree-2-swap}
\end{figure}

\begin{lemma}[Sweep]
\label{lem:el-sweep}
In the left-shelling pseu\-do-tri\-an\-gu\-la\-tion, we can interchange the labels of any number of internal top edges and their corresponding bottom edges with $O(n)$ exchanging flips.
\end{lemma}
\begin{proof}
Let $S$ be the set of vertices whose internal top edge should have its label swapped with the corresponding bottom edge. Consider a ray $L$ from $v_0$ that starts at the positive $x$-axis and sweeps through the point set to the negative $x$-axis. We will maintain the following invariant: the graph induced by the vertices to the left of $L$ is their left-shelling pseu\-do-tri\-an\-gu\-la\-tion and the graph induced by the vertices to the right of $L$ is their right-shelling pseu\-do-tri\-an\-gu\-la\-tion (both groups include $v_0$). Furthermore, the labels of the top edges of the vertices in $S$ to the right of $L$ have been interchanged with their respective bottom edges. This invariant is satisfied at the start.

Suppose that $L$ is about to pass a vertex $v_k$. If $v_k$ is on the convex hull, its top edge is not internal and no action is required for the invariant to hold after passing $v_k$. So assume that $v_k$ is not on the convex hull and consider its incident edges. It is currently part of the left-shelling pseu\-do-tri\-an\-gu\-la\-tion of points to the left of $L$, where it is the last vertex. Thus, $v_k$ is connected to $v_0$ and to one vertex to its left. It is not connected to any vertex to its right, since there are $2n - 3$ edges in total, and the left- and right-shelling pseu\-do-tri\-an\-gu\-la\-tions to each side of $L$ contribute $2(k + 1) - 3 + 2(n - k) - 3 = 2n - 4$ edges. So the only edge that crosses $L$ is an edge of the convex hull. Therefore $v_k$ has degree two, which means that we can use Lemma~\ref{lem:el-degree-2-swap} to swap the labels of its top and bottom edge with three flips if $v_k \in S$.

Furthermore, the sides of the pseudo-triangle that remains if we were to remove $v_k$, form part of the convex hull of the points to either side of $L$. Thus, flipping the top edge of $v_k$ results in the tangent from $v_k$ to the convex hull of the points to the right of $L$ -- exactly the edge needed to add $v_k$ to their right-shelling pseu\-do-tri\-an\-gu\-la\-tion. Therefore we only need $O(1)$ flips to maintain the invariant when passing $v_k$.

At the end, we have constructed the right-shelling pseu\-do-tri\-an\-gu\-la\-tion and swapped the desired edges. An analogous transformation without any swapping can transform the graph back into the left-shelling pseu\-do-tri\-an\-gu\-la\-tion with $O(n)$ flips in total.
\end{proof}

\begin{figure}[htb]
 \centering
 \includegraphics{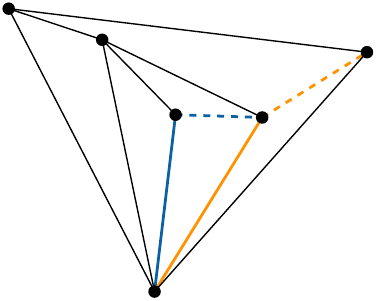}
 \caption{A pseudo-pentagon with four bitangents. It is impossible to swap the two diagonals without flipping an edge of the pseudo-pentagon, as they just flip back and forth between the solid bitangents and the dotted ones, regardless of the position of the other diagonal.}
 \label{fig:el-four-bitangents}
\end{figure}

\begin{figure}[htb]
 \centering
 \includegraphics{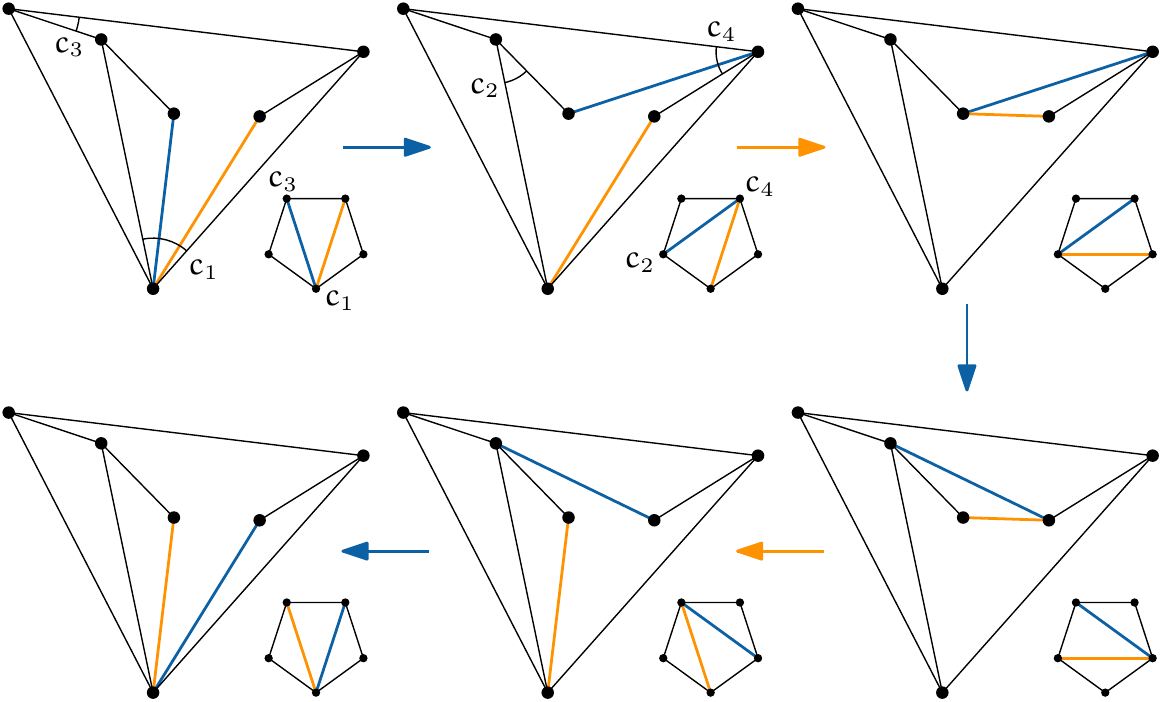}
 \caption{Interchanging the labels of two bitangents of a pseudo-pen\-tagon with five bitangents. An edge in the pentagon corresponds to a geodesic between two corners of the pseudo-pentagon.}
 \label{fig:el-pentagon-swap}
\end{figure}

\begin{lemma}
\label{lem:el-consecutive-swap}
In the left-shelling pseu\-do-tri\-an\-gu\-la\-tion, we can interchange the labels of two consecutive bottom edges with $O(1)$ exchanging flips.
\end{lemma}
\begin{proof}
When we remove the two consecutive bottom edges (say $a$ and $b$), we are left with a pseudo-pentagon $X$. A pseudo-pentagon can have up to five bitangents, as each bitangent corresponds to a geodesic between two corners. If $X$ has exactly five bitangents, this correspondence is a bijection. This implies that the bitangents of $X$ can be swapped just like diagonals of a convex pentagon (see Figure~\ref{fig:el-pentagon-swap}). On the other hand, if $X$ has only four bitangents, it is impossible to swap $a$ and $b$ without flipping an edge of $X$ (see~Figure~\ref{fig:el-four-bitangents}).

Fortunately, we can always transform $X$ into a pseudo-pentagon with five bitangents. If the pseudo-triangle to the right of $b$ is a triangle, $X$ already has five bitangents (see~Lemma~\ref{lem:el-five-bitangents-a} in Section~\ref{sec:el-deferred}). Otherwise, the top endpoint of $b$ is an internal vertex of degree two and we can flip its top edge to obtain a new pseudo-pentagon that does have five bitangents (see~Lemma~\ref{lem:el-five-bitangents-b} in Section~\ref{sec:el-deferred}). After swapping the labels of $a$ and $b$, we can flip this top edge back. Thus, in either case we can interchange the labels of $a$ and $b$ with $O(1)$ flips.
\end{proof}

We can use Lemma~\ref{lem:el-consecutive-swap} to reorder the labels of the bottom edges with insertion or bubble sort, as these algorithms only swap adjacent values.

\begin{corollary}[Shuffle]
\label{cor:el-shuffle}
In the left-shelling pseu\-do-tri\-an\-gu\-la\-tion, we can reorder the labels of all bottom edges with $O(n^2)$ exchanging flips.
\end{corollary}

Combining this with Lemmas~\ref{lem:el-sort-left-shelling} and \ref{lem:el-sweep}, and the fact that we can transform any pointed pseu\-do-tri\-an\-gu\-la\-tion into the left-shelling one with $O(n \log n)$ flips~\cite{bereg2004transforming}, gives the main result.

\begin{theorem}
\label{thm:el-upper-bound}
We can transform any edge-la\-belled pointed pseu\-do-tri\-an\-gu\-la\-tion with $n$ vertices into any other with $O(n^2)$ exchanging flips.
\end{theorem}

The following lower bound follows from the $\Omega(n \log n)$ lower bound on the flip distance between edge-la\-belled triangulations of a convex polygon~\cite{bose2013flipping}.

\begin{theorem}
\label{thm:el-lower-bound}
There are pairs of edge-la\-belled pointed pseu\-do-tri\-an\-gu\-la\-tions with $n$ vertices that require $\Omega(n \log n)$ exchanging flips to transform one into the other.
\end{theorem}
\begin{proof}
 When all vertices are in convex position, each vertex is pointed in the outer face. As such, any triangulation of these vertices is, in fact, a pointed pseudo-triangulation. Additionally, there is no difference between the standard definition of a flip in a triangulation of a convex polygon and the flip in that same setting, when interpreted as a pointed pseudo-triangulation. Thus, the lower bound of $\Omega(n \log n)$ flips to transform between two edge-labelled triangulations of a convex polygon applies to pointed pseudo-triangulations as well.
\end{proof}

\subsection{Deferred proofs}
\label{sec:el-deferred}

This section contains a few technical lemmas that were omitted from the previous section.

\begin{figure}[htb]
 \centering
 \begin{subfigure}[b]{0.48\textwidth}
  \centering
  \includegraphics{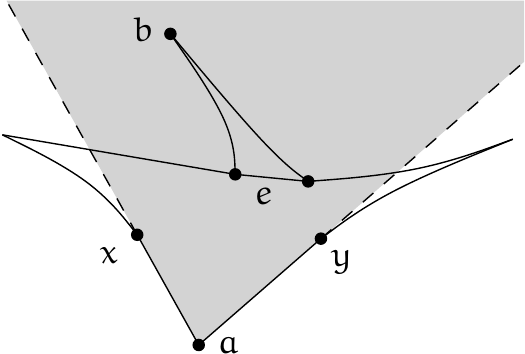}
  \caption{}
  \label{fig:el-opposite-flip}
 \end{subfigure}
 \begin{subfigure}[b]{0.48\textwidth}
  \centering
  \includegraphics{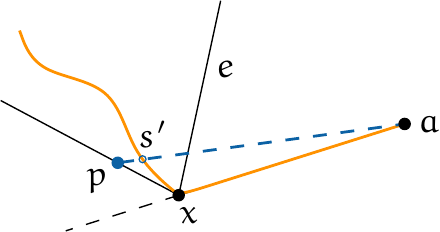}
  \caption{}
  \label{fig:el-shorter-path}
 \end{subfigure}
 \caption{(a) A corner of a pseudo-triangle and an edge such that the entire pseudo-triangle on the other side of the edge lies inside the corner's wedge. (b) If $a$ can see a point past $x$, then the geodesic does not contain $x$.}
\end{figure}

\begin{lemma}
\label{lem:el-opposite-flip}
Let $a$ be a corner of a pseudo-triangle with neighbours $x$ and $y$, and let $e$ be an edge on the chain opposite $a$. If all vertices of the other pseudo-triangle containing $e$ lie in the wedge formed by extending the edges $ax$ and $ay$ into half-lines (see Figure~\ref{fig:el-opposite-flip}), then flipping $e$ will result in an edge incident on $a$.
\end{lemma}
\begin{proof}
Let $T$ be the pseudo-triangle on the other side of $e$, and let $b$ be the corner of $T$ opposite $e$. Then flipping $e$ inserts the geodesic between $a$ and $b$. This geodesic must intersect $e$ in a point $s$ and then follow the shortest path from $s$ to $a$. If $s$ lies strictly inside the wedge, nothing can block $as$, thus the new edge will contain $as$ and be incident on $a$.

Now, if all of $e$ lies strictly inside the wedge, our result follows. But suppose that $e$ has $x$ as an endpoint and the geodesic between $a$ and $b$ intersects $e$ in $x$. As $a$ can see $x$ and all of $T$ lies inside the wedge, there is an $\varepsilon > 0$ such that $a$ can see the point $X$ on the boundary of $T$ at distance $\varepsilon$ from $x$ (see Figure~\ref{fig:el-shorter-path}). The line segment $ap$ intersects the geodesic at a point $s'$. By the triangle inequality, $s'a$ is shorter than following the geodesic from $s'$ via $x$ to $a$. But then this would give a shorter path between $a$ and $b$, by following the geodesic to $s'$ and then cutting directly to $a$. As the geodesic is the shortest path by definition, this is impossible. Thus, the geodesic cannot intersect $e$ at $x$ and the new edge must be incident to $a$.
\end{proof}

\begin{lemma}
\label{lem:el-five-bitangents-a}
Let $a$ and $b$ be two consecutive internal bottom edges in the left-shelling pseu\-do-tri\-an\-gu\-la\-tion, such that the pseudo-triangle to the right of $b$ is a triangle. Then the pseudo-pentagon $X$ formed by removing $a$ and $b$ has five bitangents.
\end{lemma}
\begin{proof}
Let $c_0, \ldots, c_4$ be the corners of $X$ in counter-clockwise order around the boundary. By Lemma~\ref{lem:el-opposite-flip}, flipping $b$ results in an edge $b'$ that intersects $b$ and is incident on $c_1$. This edge is part of the geodesic between $c_1$ and $c_3$, and as such it is tangent to the convex chain $v_0, v_a, \ldots, c_3$, where $v_a$ is the top endpoint of $a$ ($v_a$ could be $c_3$). Therefore it is also the tangent from $c_1$ to the convex hull of $\{v_0, \ldots, v_a\}$. This means that the newly created pseudo-triangle with $c_1$ as corner and $a$ on the opposite pseudo-edge also meets the conditions of Lemma~\ref{lem:el-opposite-flip}. Thus, flipping $a$ results in another edge, $a'$, also incident on $c_1$. As $b$ separates $c_1$ from all vertices in $\{v_0, \ldots, v_a\}$, $a'$ must also intersect $b$. This gives us four bitangents, of which two are incident on $v_0$ ($a$ and $b$), and two on $c_1$ ($a'$ and $b'$). Finally, flipping $a$ before flipping $b$ results in a bitangent that is not incident on $v_0$ (as $v_0$ is a corner and cannot be on the new geodesic), nor on $c_1$ (as $b$ separates $a$ from $c_1$). Thus, $X$ has five bitangents.
\end{proof}

\begin{lemma}
\label{lem:el-five-bitangents-b}
Let $a$ and $b$ be two consecutive internal bottom edges in the left-shelling pseu\-do-tri\-an\-gu\-la\-tion, such that the pseudo-triangle to the right of $b$ is not a triangle. Then the pseudo-pentagon $X$ formed by flipping the corresponding top edge of $b$ and removing $a$ and $b$ has five bitangents.
\end{lemma}
\begin{proof}
Let $v_a$ and $v_b$ be the top endpoints of $a$ and $b$. By Lemma~\ref{lem:el-opposite-flip} and since $b$ had degree two, flipping the top edge of $b$ results in the edge $v_bc_1$. We get three bitangents for free: $a$, $b$, and $b'$ -- the old top edge of $b$ and the result of flipping $b$.

$X$ consists of a reflex chain $C$ that is part of the convex hull of the points to the left of $a$, followed by three successive tangents to $C$, $v_a$, or $v_b$. Since $C$ lies completely to the left of $a$, it cannot significantly alter any of the geodesics or bitangents inside the polygon, so we can reduce it to a single edge. Now, $X$ consists either of a triangle with two internal vertices, or a convex quadrilateral with one internal vertex.

If $X$ is a triangle with two internal vertices, the internal vertices are $v_a$ and $v_b$. Let its exterior vertices be $v_0$, $x$, and $y$. Then there are seven possible bitangents: $a = v_0v_a$, $b = v_0v_b$, $xv_a$, $xv_b$, $yv_a$, $yv_b$, and $v_av_b$. We know that $xv_a$ and $yv_b$ are edges, so there are five possible bitangents left. As all vertices involved are either corners or have degree one in $X$, the only condition for an edge to be a bitangent is that it does not cross the boundary of $X$. Since the exterior boundary is a triangle, this reduces to it not crossing $xv_a$ and $yv_b$. Two line segments incident to the same vertex cannot cross. Thus, $xv_b$, $yv_a$, and $v_av_b$ cannot cross $xv_a$ and $yv_b$, and $X$ has five bitangents.

If $X$'s convex hull has four vertices, the internal vertex is $v_b$ (otherwise the pseudo-triangle to the right of $b$ would be a triangle). Let its exterior vertices be $v_0$, $x$, $v_a$, and $y$. Then there are six possible bitangents: $a = v_0v_a$, $b = v_0v_b$, $xy$, $xv_b$, $yv_b$, and $v_av_b$, of which one ($yv_b$) is an edge of $X$. Since $a$ and $b$ are guaranteed to be bitangents, and $xy$, $xv_b$, and $v_av_b$ all share an endpoint with $yv_b$, the arguments from the previous case apply and we again have five bitangents.
\end{proof}

\section{General pseu\-do-tri\-an\-gu\-la\-tions}
\label{sec:el-general-pts}

In this section, we extend our results for edge-la\-belled pointed pseu\-do-tri\-an\-gu\-la\-tions to all edge-la\-belled pseu\-do-tri\-an\-gu\-la\-tions. Since not all pseu\-do-tri\-an\-gu\-la\-tions have the same number of edges, we need to allow flips that change the number of edges. In particular, we allow a single edge to be deleted or inserted, provided that the result is still a pseu\-do-tri\-an\-gu\-la\-tion.

Since we are dealing with edge-la\-belled pseu\-do-tri\-an\-gu\-la\-tions, we need to determine what happens to the edge labels. It is useful to first review the properties we would like these flips to have. First, a flip should be a local operation -- it should affect only one edge. Second, a labelled edge should be flippable if and only if the edge is flippable in the unlabelled setting. This allows us to re-use the existing results on flips in pseu\-do-tri\-an\-gu\-la\-tions. Third, flips should be reversible. Like most proofs about flips, our proof in the previous section crucially relies on the reversibility of flips.

With these properties in mind, the edge-deletion flip is rather straightforward -- the labelled edge is removed, and other edges are not affected. Since the edge-insertion flip needs to be the inverse of this, it should insert the edge and assign it a \emph{free label} -- an unused label in $\{1, \ldots, 3n - 3 - 2h\}$, where $h$ is the number of vertices on the convex hull ($3n - 3 - 2h$ is the number of internal edges in a triangulation).

With the definitions out of the way, we can turn our attention to the number of flips required to transform any edge-la\-belled pseu\-do-tri\-an\-gu\-la\-tions into any other. In this section, we show that by using insertion and deletion flips, we can shuffle (permute the labels on bottom edges) with $O(n + h \log h)$ flips. Combined with the unlabelled bound of $O(n \log c)$ flips by Aichholzer~\etal~\cite{aichholzer2006transforming}, this brings the total number of flips down to $O(n \log c + h \log h)$. When the points are in convex position ($h = n$), all pseudo-triangulations are full triangulations, and the $O(n \log n)$ bound by Bose~\etal~\cite{bose2013flipping} in this setting implies our bound. Therefore, we assume for the remainder of this section that the points are not in convex position ($h < n$). As before, we first build a collection of simple tools that help prove the main result.

\begin{lemma}
 \label{lem:el-degree-2-free}
 With $O(1)$ flips, we can interchange the label of an edge incident to an internal vertex of degree two with a free label.
\end{lemma}
\begin{proof}
 Let $v$ be a vertex of degree two and let $e$ be an edge incident to $v$. Since $v$ has degree two, its removal leaves an empty pseudo-triangle $T$. There are three bitangents that connect $v$ to $T$, one for each corner. Thus, we can insert the third bitangent $f$ with the desired free label, making $v$ non-pointed (see~Figure~\ref{fig:el-degree-2-free}). Flipping $e$ now removes it and frees its label. Finally, flipping $f$ moves it into $e$'s starting position, completing the exchange.
\end{proof}

\begin{figure}[htb]
 \centering
 \includegraphics{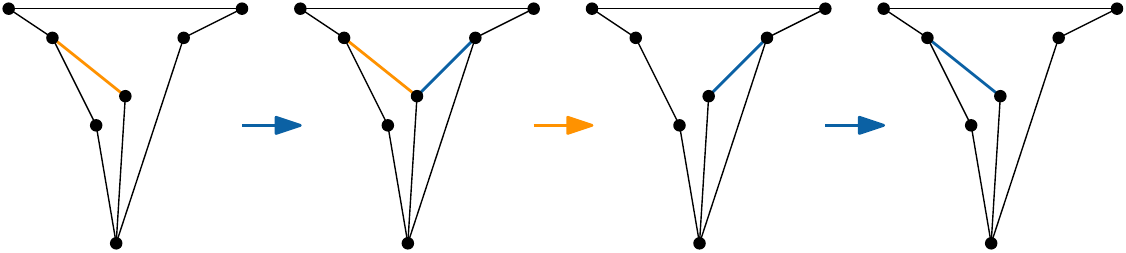}
 \caption{Interchanging the label of an edge incident to a vertex of degree two with a free label.}
 \label{fig:el-degree-2-free}
\end{figure}

This implies that, using an arbitrary free label as placeholder, we can swap any two edges incident to internal degree-two vertices -- no matter where they are in the pseu\-do-tri\-an\-gu\-la\-tion.

\begin{corollary}
 \label{lem:el-two-degree-2-swap}
 We can interchange the labels of two edges, each incident to some internal vertex of degree two, with $O(1)$ flips.
\end{corollary}

Recall that during a sweep (Lemma~\ref{lem:el-sweep}), each internal vertex has degree two at some point. Since the number of free labels for a pointed pseu\-do-tri\-an\-gu\-la\-tion is equal to the number of internal vertices, this means that we can use Lemma~\ref{lem:el-degree-2-free} to swap every label on a bottom edge incident to an internal vertex with a free label by performing a single sweep. Afterwards, a second sweep can replace these labels on the bottom edges in any desired order. Thus, permuting the labels on bottom edges incident to internal vertices can be done with $O(n)$ flips. Therefore, the difficulty in permuting the labels on all bottom edges lies in bottom edges that are not incident to an internal vertex, that is, chords of the convex hull. If there are few such chords, a similar strategy (free them all and replace them in the desired order) might work. Unfortunately, the number of free labels can be far less than the number of chords.

We now consider operations on maximal groups of consecutive chords, which we call \emph{fans}. As the vertices of a fan are in convex position, fans behave in many ways like triangulations of a convex polygon, which can be rearranged with $O(n \log n)$ flips~\cite{bose2013flipping}. The problem now becomes getting the right set of labels on the edges of a fan.

Consider the internal vertices directly to the left ($\vl$) and right ($\vr$) of a fan $F$, supposing both exist. Vertex $\vl$ has degree two and forms part of the reflex chain of the first pseudo-triangle to the left of $F$. Thus, flipping $\vl$'s top edge connects it to the leftmost vertex of $F$ (excluding $v_0$). Vertex $\vr$ is already connected to the rightmost vertex of $F$, so we just ensure that it has degree two. To do this, we flip all incident edges from vertices further to the right, from the bottom to the top. Now the diagonals of $F$ form a triangulation of a convex polygon whose boundary consists of $v_0$, $\vl$, the top endpoints of the chords, and $\vr$ (see~Figure~\ref{fig:el-indexed-fan}). It is possible that there is no internal vertex to one side of $F$. In that case, there is only one vertex on that side of $F$, which is part of the convex hull, and we can simply use that vertex in place of $\vl$ or $\vr$ without flipping any of its edges. Since there is at least one internal vertex by assumption, either $\vl$ or $\vr$ is an internal vertex. This vertex is called the \emph{index} of $F$. If a vertex is the index of two fans, it is called a \emph{shared index}.

\begin{figure}[htb]
 \centering
 \begin{subfigure}[b]{0.38\textwidth}
  \centering
  \includegraphics{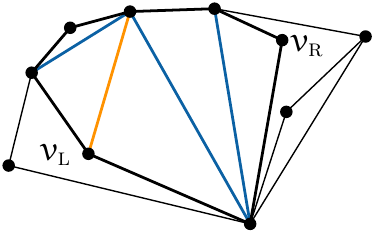}
  \caption{}
  \label{fig:el-indexed-fan}
 \end{subfigure}
 \begin{subfigure}[b]{0.58\textwidth}
  \centering
  \includegraphics{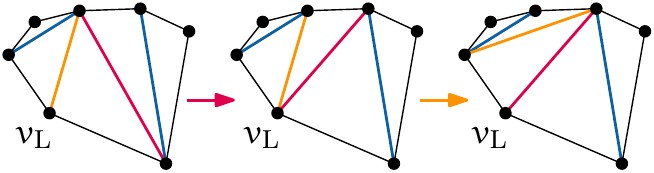}
  \caption{}
  \label{fig:el-increase-index}
 \end{subfigure}
 \caption{(a) An indexed fan. (b) Shifting the index ($\vl$) from the yellow edge to the red edge.}
\end{figure}

A triangulated fan is called an \emph{indexed fan} if there is one edge incident to the index, the \emph{indexed edge}, and the remaining edges are incident to one of the neighbours of the index on the boundary. Initially, all diagonals of $F$ are incident to $v_0$, so we transform it into an indexed fan by flipping the diagonal of $F$ closest to the index. Next, we investigate several operations on indexed fans that help us move labels between fans.
  
\begin{lemma}[Shift]
 In an indexed fan, we can shift the indexed edge to the next diagonal with $O(1)$ flips.
\end{lemma}
\begin{proof}
 Suppose that $\vl$ is the index (the proof for $\vr$ is analogous). Let $e$ be the current indexed edge, and $f$ be the leftmost diagonal incident to $v_0$. Then flipping $f$ followed by $e$ makes $f$ the only edge incident to the index and $e$ incident to the neighbour of the index (see~Figure~\ref{fig:el-increase-index}). Since flips are reversible, we can shift the index the other way too.
\end{proof}

\begin{figure}[htb]
 \centering
 \includegraphics{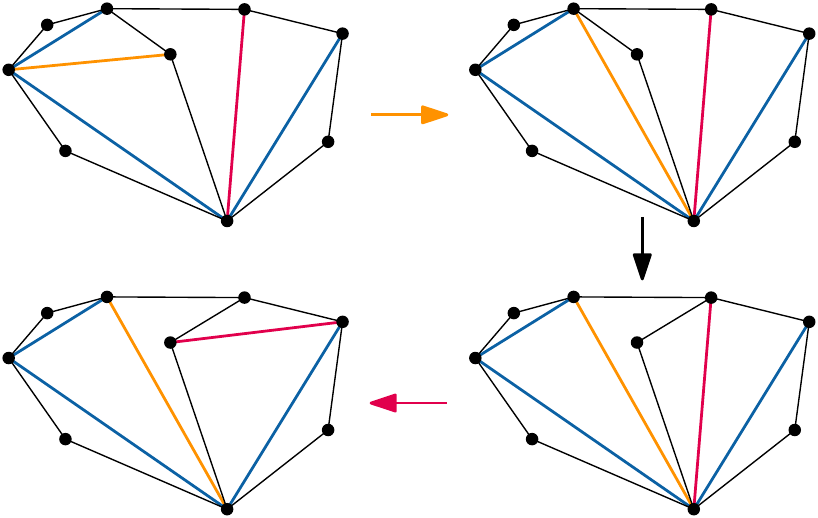}
 \caption{Changing which side a shared index indexes.}
 \label{fig:el-swap-shared-index}
\end{figure}

\begin{lemma}
 We can switch which fan a shared index currently indexes with $O(1)$ flips.
\end{lemma}
\begin{proof}
 Flipping the current indexed edge ``parks" it by connecting it to the two neighbours of the index, and reduces the degree of the index to two (see~Figure~\ref{fig:el-swap-shared-index}). Now, flipping the top edge of the index connects it to the other fan, where we parked the previously indexed edge. Flipping that edge connects it to the index again.
\end{proof}

\begin{lemma}
 \label{lem:el-deg-3-to-2}
 In a pointed pseu\-do-tri\-an\-gu\-la\-tion, we can always decrease the degree of a vertex $v$ of degree three by flipping one of the edges incident to its reflex angle.
\end{lemma}
\begin{proof}
 Consider the geodesic from $v$ to the opposite corner $c$ of the pseudo-triangle $v$ is pointed in. The line supporting the part of the geodesic when it reaches $v$ splits the edges incident to $v$ into two groups. As there are three edges, one of these groups must contain multiple edges. Flipping the edge incident to its reflex angle in the group with multiple edges results in a geodesic to $c$. If this geodesic passed through $v$, it would insert the missing edges along the geodesic from $v$ to $c$ (otherwise we could find a shorter path). But inserting this geodesic would make $v$ non-pointed. Thus, $v$ cannot be on this geodesic. Therefore the new edge is not incident to $v$ and the flip reduces the degree of $v$.
\end{proof}

Since the index always has degree three, this allows us to extend the results from Lemmas~\ref{lem:el-degree-2-free} and \ref{lem:el-degree-2-swap} regarding vertices of degree two to indexed edges.

\begin{corollary}
 \label{cor:el-index-free}
 In an indexed fan, we can interchange the label of the indexed edge with a free label in $O(1)$ flips.
\end{corollary}

\begin{corollary}
 \label{cor:el-index-swap}
 Given two indexed fans, we can interchange the labels of the two indexed edges with $O(1)$ flips.
\end{corollary}

Now we have enough tools to shuffle the bottom edges.
 
\begin{lemma}[Shuffle]
\label{lem:el-shuffle-non-pointed}
 In the left-shelling pseu\-do-tri\-an\-gu\-la\-tion, we can reorder the labels of all bottom edges with $O(n + h \log h)$ flips, where $h$ is the number of vertices on the convex hull.
\end{lemma}
\begin{proof}
 In the initial pseu\-do-tri\-an\-gu\-la\-tion, let $B$ and $\free$ be the sets of labels on bottom edges and free labels, respectively. Let $F_i$ be the set of labels on the $i$-th fan (in some fixed order), and let $\nf$ be the set of labels on non-fan bottom edges. Let $F_i'$ and $\nf'$ be these same sets in the target pseu\-do-tri\-an\-gu\-la\-tion. As we are only rearranging the bottom labels, we have that $B = F_1 \cup \ldots \cup F_k \cup \nf = F_1' \cup \ldots \cup F_k' \cup \nf'$, where $k$ is the number of fans.

 We say that a label $\ell$ \emph{belongs to} fan $i$ if $\ell \in F_i'$. At a high level, the reordering proceeds in four stages. In stage one, we free all labels in $\nf$. In stage two, we place each label from $B \setminus \nf'$ in the fan it belongs to, leaving the labels in $\nf'$ free. Then, in stage three, we correct the order of the labels within each fan. Finally, we place the labels in $\nf'$ correctly.

 Since each internal vertex contributes exactly one top edge, one bottom edge, and one free label, we have that $|\nf| = |\free|$. To free all labels in $\nf$, we perform a sweep (see~Lemma~\ref{lem:el-sweep}). As every internal vertex has degree two at some point during the sweep, we can exchange the label on its bottom edge with a free label at that point, using Lemma~\ref{lem:el-degree-2-free}. This requires $O(n)$ flips. The labels in $\free$ remain on the bottom edges incident to internal vertices throughout stage two and three, as placeholders.
 
 To begin stage two, we index all fans with $O(n)$ flips and shift these indices to the first `foreign' edge: the first edge whose label does not belong to the current fan. If no such edge exists, we can ignore this fan for the remainder of stage two, as it already has the right set of labels. Now suppose that there is a fan $F_i$ whose indexed edge $e$ is foreign: $\ell_e \notin F_i'$. Then either $\ell_e \in F_j'$ for some $j \neq i$, or $\ell_e \in \nf'$. In the first case, we exchange $\ell_e$ with the label on the indexed edge of $F_j$, and shift the index of $F_j$ to the next foreign edge. In the second case, we exchange $\ell_e$ with a free label in $B \setminus \nf'$. If this label belongs to $F_i$, we shift its index to the next foreign edge. In either case, we increased the number of correctly placed labels by at least one. Thus $n - 1$ repetitions suffice to place all labels in the fan they belong to, wrapping up stage two. Since we perform a linear number of swaps and shifts, and each takes a constant number of flips, the total number of flips required for stage two is $O(n)$.

 For stage three, we note that each indexed fan corresponds to a triangulation of a convex polygon. As such, we can rearrange the labelled diagonals of a fan $F_i$ into their desired final position with $O(|F_i| \log |F_i|)$ flips~\cite{bose2013flipping}. Thus, if we let $h$ be the number of vertices on the convex hull, the total number of flips for this step is bounded by
 \[ 
  \sum_i O(|F_i| \log |F_i|)~~\leq~~\sum_i O(|F_i| \log h)~~=~~O(h \log h).
 \]
 For stage four, we first return to a left-shelling pseu\-do-tri\-an\-gu\-la\-tion by un-indexing each fan, using $O(n)$ flips. After stage two, the labels in $\nf'$ are all free, so all that is left is to place these on the correct bottom edges, which we can do with a final sweep. Thus, we can reorder all bottom labels with $O(n + h \log h)$.
\end{proof}

This leads to the following bound.

\begin{theorem}
 We can transform any edge-la\-belled pseu\-do-tri\-an\-gu\-la\-tion with $n$ vertices into any other with $O(n \log c + h \log h)$ flips, where $c$ is the number of convex layers and $h$ is the number of vertices on the convex hull.
\end{theorem}
\begin{proof}
 Using the technique by Aichholzer~\etal~\cite{aichholzer2006transforming}, we first transform the pseu\-do-tri\-an\-gu\-la\-tion into the left-shelling pseu\-do-tri\-an\-gu\-la\-tion $T$ with $O(n \log c)$ flips. Our canonical pseu\-do-tri\-an\-gu\-la\-tion contains the labels $\{1, \ldots, 2n - h - 3\}$, but it is possible for $T$ to contain a different set of labels. Since all labels are drawn from $\{1, \ldots, 3n - 2h - 3\}$, at most $n - h$ labels differ. This is exactly the number of internal vertices. Thus, we can use $O(n + h \log h)$ flips to shuffle (Lemma~\ref{lem:el-shuffle-non-pointed}) all non-canonical labels on fan edges to bottom edges incident to an internal vertex. Once there, we use a sweep (Lemma~\ref{lem:el-sweep}) to ensure that every internal vertex has degree two at some point, at which time we replace its incident non-canonical labels with canonical ones with a constant number of flips (Lemma~\ref{lem:el-degree-2-free}). Once our left-shelling pseu\-do-tri\-an\-gu\-la\-tion has the correct set of labels, we use a constant number of shuffles and sweeps to sort the labels (Lemma~\ref{lem:el-sort-left-shelling}). Since we can shuffle and sweep with $O(n + h \log h)$ and $O(n)$ flips, respectively, the total number of flips reduces to $O(n \log c + n + h \log h) = O(n \log c + h \log h)$.
\end{proof}

The correspondence between triangulations of a convex polygon and pseu\-do-tri\-an\-gu\-la\-tions, proven in Theorem~\ref{thm:el-lower-bound}, also gives us the following lower bound, as no insertion or deletion flips are possible in this setting.

\begin{theorem}
\label{thm:el-non-pointed-lower-bound}
There are pairs of edge-la\-belled pseu\-do-tri\-an\-gu\-la\-tions with $n$ vertices such that any sequence of flips that transforms one into the other has length $\Omega(n \log n)$.
\end{theorem}

\section*{Acknowledgments}

We would like to thank Anna Lubiw and Vinayak Pathak for helpful discussions.

\bibliography{references}

\begin{thebibliography}{10}
\expandafter\ifx\csname url\endcsname\relax
  \def\url#1{\texttt{#1}}\fi
\expandafter\ifx\csname urlprefix\endcsname\relax\def\urlprefix{URL }\fi
\expandafter\ifx\csname href\endcsname\relax
  \def\href#1#2{#2} \def\path#1{#1}\fi

\bibitem{kirkpatrick2002kinetic}
D.~Kirkpatrick, J.~Snoeyink, B.~Speckmann, Kinetic collision detection for
  simple polygons, International Journal of Computational Geometry \&
  Applications 12~(1-2) (2002) 3--27.
\newblock \href {http://dx.doi.org/10.1142/S0218195902000724}
  {\path{doi:10.1142/S0218195902000724}}.

\bibitem{streinu2005pseudo}
I.~Streinu, Pseudo-triangulations, rigidity and motion planning, Discrete \&
  Computational Geometry. 34~(4) (2005) 587--635.
\newblock \href {http://dx.doi.org/10.1007/s00454-005-1184-0}
  {\path{doi:10.1007/s00454-005-1184-0}}.

\bibitem{rote2007pseudotriangulations}
G.~Rote, F.~Santos, I.~Streinu, Pseudo-triangulations---a survey, in: Surveys
  on discrete and computational geometry, Vol. 453 of Contemporary Mathematics,
  2008, pp. 343--410.
\newblock \href {http://dx.doi.org/10.1090/conm/453/08807}
  {\path{doi:10.1090/conm/453/08807}}.

\bibitem{lawson1972transforming}
C.~L. Lawson, Transforming triangulations, Discrete Mathematics 3~(4) (1972)
  365--372.

\bibitem{hurtado1999flipping}
F.~Hurtado, M.~Noy, J.~Urrutia, Flipping edges in triangulations, Discrete \&
  Computational Geometry. 22~(3) (1999) 333--346.
\newblock \href {http://dx.doi.org/10.1007/PL00009464}
  {\path{doi:10.1007/PL00009464}}.

\bibitem{bereg2004transforming}
S.~Bereg, Transforming pseudo-triangulations, Information Processing Letters
  90~(3) (2004) 141--145.
\newblock \href {http://dx.doi.org/10.1016/j.ipl.2004.01.021}
  {\path{doi:10.1016/j.ipl.2004.01.021}}.

\bibitem{aichholzer2003pseudotriangulations}
O.~Aichholzer, F.~Aurenhammer, H.~Krasser, P.~Brass, Pseudotriangulations from
  surfaces and a novel type of edge flip, SIAM Journal on Computing 32~(6)
  (2003) 1621--1653 (electronic).
\newblock \href {http://dx.doi.org/10.1137/S0097539702411368}
  {\path{doi:10.1137/S0097539702411368}}.

\bibitem{aichholzer2006transforming}
O.~Aichholzer, F.~Aurenhammer, C.~Huemer, H.~Krasser, Transforming spanning
  trees and pseudo-triangulations, Information Processing Letters 97~(1) (2006)
  19--22.
\newblock \href {http://dx.doi.org/10.1016/j.ipl.2005.09.003}
  {\path{doi:10.1016/j.ipl.2005.09.003}}.

\bibitem{araujo2015colorful}
G.~Araujo-Pardo, I.~Hubard, D.~Oliveros, E.~Schulte, Colorful associahedra and
  cyclohedra, Journal of Combinatorial Theory, Series A 129 (2015) 122--141.
\newblock \href {http://dx.doi.org/10.1016/j.jcta.2014.09.001}
  {\path{doi:10.1016/j.jcta.2014.09.001}}.

\bibitem{bose2013flipping}
P.~Bose, A.~Lubiw, V.~Pathak, S.~Verdonschot, Flipping edge-labelled
  triangulations (2013).
\newblock \href {http://arxiv.org/abs/1310.1166} {\path{arXiv:1310.1166}}.

\end{thebibliography}

\end{document}